\documentclass{llncs}
\usepackage{amsmath}
\usepackage{amssymb}
\usepackage{amsfonts}
\usepackage{color}
\usepackage{makeidx}
\usepackage{hyperref}
\usepackage{xspace}

\newcommand{\blue}[1]{\textcolor{black}{{#1}}}

\newcommand{\ignore}[1]{}

\newcommand{\cB}{{\cal B}}
\newcommand{\cF}{{\cal F}}

\renewcommand{\H}{{\cal H}}

\newcommand{\oftypenon}[2]{\set{\,{#1}\to{#2}\,}}

\newcommand{\oftype}[2]{\subseteq \oftypenon{#1}{#2}}

\newcommand{\B}{\{0,1\}}
\newcommand{\cff}{$(n,(r,s))$-CFF\xspace}

\newcommand{\N}{$N(n,(r,s))$\xspace}
\newcommand{\set}[1]{\{#1\}}
\newcommand{\rsimp}{\textsc{$r$-Simple $k$-Path}\xspace}
\def\Z{{\mathbb{Z}}}

\newcommand{\F}{{\cal F}}
\renewcommand{\P}{{\cal P}}

\newcommand{\rset}{$r$-\textrm{set}\xspace}
\newcommand{\rsets}{$r$-\textrm{set}s\xspace}
\newcommand{\rkset}{$(r,k)$-\textrm{set}\xspace}
\newcommand{\rksets}{$(r,k)$-\textrm{set}s\xspace}
\newcommand{\rksep}{$(r,k)$-separator\xspace}

\newcommand{\repcomputenon}{r^{8k/r + o(k/r)}\cdot 2^{O(k/r)}} 

\newcommand{\mondetectcomputenon}{r^{12k/r+o(k/r)}\cdot 2^{O(k/r)}} 

\newcommand{\sepcompute}{r^{4k/r+o(r)}\cdot 2^{O(k/r)}\cdot n\log n} 

\newcommand{\comp}[1]{\overline{#1}}

\newcommand{\rmondetect}{\textsc{$(r,k)$-Monomial Detection}\xspace}

\newcommand{\rr}{[r]_0}

\begin{document}
\title{Almost Optimal Cover-Free Families}
\author{Nader H. Bshouty\inst{}
\and Ariel Gabizon\inst{}}
\institute{Department of Computer Science\\ Technion, Haifa, 32000 }


%
\maketitle

\begin{abstract}
Roughly speaking, an \emph{$(n,(r,s))$-Cover Free Family} (CFF) is a small set of $n$-bit strings
such that: ``in any $d:=r+s$ indices we see all patterns of weight $r$''.
CFFs have been of interest for a long time both in discrete mathematics
as part of block design theory, and in theoretical computer science where
they have found a variety of applications, for example, in parametrized algorithms where they were introduced in the recent breakthrough work of Fomin, Lokshtanov and Saurabh
\cite{FLS14} under the name `lopsided universal sets'.

In this paper we give the first explicit construction of
cover-free families of optimal size up to lower order multiplicative terms,
\emph{for any $r$ and $s$}.
In fact, our construction time is almost linear in the size of the family.
Before our work, such a result existed only for {$r=d^{o(1)}$}.
and {$r= \omega(d/(\log\log d\log\log\log d))$}.
As a sample application, we improve the running times of parameterized algorithms from
the recent work of Gabizon, Lokshtanov and Pilipczuk \cite{GLP}.
\end{abstract}

\section{Introduction}
The purpose of this paper is to give an explicit almost optimal construction
of \emph{cover free families}~\cite{KS64}.
Before giving a formal definition, let us describe the special
case of \emph{group testing}.
The problem of  group testing was first presented during
World War II and described as follows~\cite{DH00,ND00}: Among $n$ soldiers, at most
$s$ carry a fatal virus. We would like to blood test the soldiers
to detect the infected ones. Testing each one separately will give
$n$ tests. To minimize the number of tests we can mix the blood of
several soldiers and test the mixture. If the test comes negative
then none of the tested soldiers are infected. If the test comes
out positive, we know that at least one of them is infected. The
problem is to come up with a small number of tests.

\blue{To obtain a non-adaptive algorithm for this problem},
a little thought shows that what is required is a set of tests such that
for any subset $T$ of $s$ soldiers, and any soldier $i\notin T$,
there is a test including soldier $i$, and precluding all soldiers in $T$.
Let $d= s+1$.
Viewing a test as a characteristic vector $a\in \B^n$ of the soldiers it includes, the desired property
is equivalent to the following. Find a small set
$\cF\subseteq\B^n$ such that for every $1\le i_1<i_2<\cdots<i_d\le
n$, and every $1\le j\le d$, there is $a\in \cF$ such that $a_{i_j}=1$
and $a_{i_k}=0$ for all $k\not=j$.

\subsection{Cover-Free Families}
We can view $\cF$ described above as a set of strings such
that ``in any $d$ indices we see all patterns of weight one''.
We can generalize this property by choosing an integer $1 \leq r <d$
and requesting to see ``in any $d$ indices all patterns of weight~$r$''.

\begin{definition}[Cover-Free Family]\label{dfn-cff}
Fix positive integers $r,s,n$ with $r,s<n$ and let $d:=r+s$.
An \emph{$(n,(r,s))$-Cover Free Family} (CFF) is a set $\cF\subseteq \B^n$ such that
 for every $1\le i_1<i_2<\cdots<i_{d}\le n$ and
every $J\subset [d]$ of size $|J|=r$ there is $a\in \cF$
 such that $a_{i_j}=1$ for $j\in  J$ and $a_{i_k}=0$ for $k\notin J$.
\end{definition}

 \emph{We will always assume $r\leq d/2$ \blue{(and therefore
 $r\le s$)}: If not, construct an $(n,(s,r))$-CFF and
 take the set of complement vectors.}

We note that the definition of CFFs usually given is a different equivalent one
which we now describe.
Given an \cff $\cF$, denote $N=|\cF|$ and construct the $N\times n$ boolean matrix $A$ whose rows
are the elements of $\cF$.
Now, let $X$ be a set of $N$ elements and think of the \emph{columns} of $A$
as characteristic vectors of subsets, which we will call \emph{blocks}, $B\subseteq X$.
That is, if we denote by $\cB=\set{B_1,\ldots,B_n}$ the set of blocks corresponding to these columns,
then $A$ is the \emph{incidence matrix} of $\cB$, i.e.
the $i$'th element of $X$ is in  $B_j$ if and only if $A_{i,j}=1$.

For this view, the CFF property of $\cF$ implies the following:
For any blocks $B_1,\ldots,B_r \in \cB$ and any other $s$ blocks
$A_1,\ldots, A_s \in \cB$ (distinct from the $B$'s), there is an element of $X$ contained in all the $B$'s but
not in any of the $A$'s, i.e.
$$\bigcap_{i=1}^r B_i\not\subseteq \bigcup_{j=1}^s A_j.$$

This property is the usual way to define CFFs ~\cite{KS64}.

\paragraph{Notation:}
Let us denote by \N the minimal integer $N$ such that there exists an \cff $\cF$ of size $|\cF| = N$.

\subsection{Previous Results}\label{OR}

It is known that, \cite{SWZ00},
\blue{
$$N(n,(r,s))\ge \Omega (N(r,s)\cdot \log n)$$
where
$$N(r,s):=\frac{d{d\choose r}}{\log{d\choose r}}.$$
}

Using the union bound it is easy to show
that for $d=r+s=o(n)$, $r\le s$, we have
$$N(n,(r,s))\le O\left(\sqrt{{r}}\log{d\choose r}\cdot N(r,s)\cdot\log n\right).$$
\blue{D'yachkov et. al.'s breakthrough result,~\cite{DVPS14}, implies that for $s,n\to\infty$
\begin{eqnarray}\label{fb2}
N(n,(r,s))= \Theta\left( N(r,s)\cdot \log n\right).
\end{eqnarray}
The two above bounds are non-constructive.}

It follows from \cite{SWZ}, that for an infinite sequence of integers $n$, an $(n,(r,s))$-CFF of size
$$M=O\left((rd)^{\log^* n}\log n\right)$$ can be constructed in polynomial time.

Before proceeding to describe previous results and ours, we introduce some convenient terminology:


We will think of the parameter $d=r+s$ as going to infinity and always use
the notation $o(1)$ for a term that is independent of $n$, and  goes to $0$ as $d\mapsto \infty$.

We say an $(n,(r,s))$-CFF $\cF$ is
{\it almost optimal},
if its size $N = |\cF| $ satisfies
\begin{eqnarray}\label{TB}
N&=& N(r,s)^{1+o(1)}\cdot \log n
= \begin{cases} d^{r+1+o(1)}\log n &\mbox{if } r=O(1) \\
\left(\frac{d}{r}\right)^{r+o(r)}\log n & \mbox{if } r=\omega(1), r=o(d)\\
2^{H_2(r/d)d+o(d)}\log n & \mbox{if } r=O(d) \end{cases} \nonumber.\end{eqnarray}
where $H_2(x)$ is the binary intopy function.

We say that such $\cF$ can be \emph{constructed in linear time} if
it can be constructed in time
$O(N(r,s)^{1+o(1)}\cdot \log n\cdot n)$.
In this terminology, our goal is to obtain almost optimal CFFs that
are constructible in linear time.

Let us first consider the case of constant $r$.
It is not hard to see that in this case an \cff $\cF$ of size
$d^{r+1}\log n$ is almost optimal by our definition \blue{(and in fact
exceeds the optimal size in (\ref{TB}) only by a multiplicative $\log d$ factor).}
Bshouty \cite{B14b} constructs $\cF$ of such size in linear time
and thus solves the case of constant $r$.
In fact, calculation shows that for any $r=d^{o(1)}$, $\cF$ of size
\[N=2^{O(r)}\cdot d^{r+1}\cdot \log n \]
is almost optimal.
Bshouty \cite{B14,B14b} constructs such $\cF$ in linear time for any $r=o(d)$.

\ignore{\begin{center}
\begin{tabular}{|c|c|c|c|c|}
\hline
 & Linear time.  & Upper & Lower \\

 $r$ & Size=$O(\ )$ & Bound & Bound \\
\hline\hline
 $O(1)$ & ${d^{r+1}}\log n$ & $\frac{d^{r+1}}{\log d}\log n$ & $\frac{d^{r+1}}{\log d}\log n$\\
\hline
 $o(d)$ & ${(ce)^rd^{r+1}}\log n$ & $\frac{d^{r+1}}{(r/e)^{r-1/2}\log d}\log n$ & $\frac{d^{r+1}}{(r/e)^{r+1}\log d}\log n$\\
\hline
\end{tabular}
\end{center}
In the table, $c>1$ is any constant and $e$ is the Euler's constant $2.718\cdots$.
}
\noindent
 We proceed to the case of larger $r$.
Fomin et. al. \cite{FLS14} construct an
\cff of size
\begin{eqnarray}\label{Fom}{d\choose r} 2^{O\left(\frac{d}{\log\log (d)}\right)}\log n\end{eqnarray}
in linear time.
This is almost optimal
when
$$r=\omega\left(\frac{d}{\log\log d\log\log\log d}\right).$$
To the best of our knowledge there is no explicit construction of almost optimal $(n,(r,s))$-CFFs
when $d^{o(1)}<r<\omega(d/(\log\log d\log\log\log d).$

Note that in this range
(and even for $r=\omega(1)$ and $r=o(d)$), $\cF$ is almost optimal if and only if it has size
\[N={d\choose r}^{1+o(1)} \log n = \left(\frac{d}{r}\right)^{r(1+o(1))}\cdot \log n.\]

Gabizon et. al \cite{GLP} made a significant step for general $r$ and  constructed an \cff
of size
\[O((d/r)^{2\cdot r}\cdot 2^{O(r)}\cdot \log n)\]
in linear time. This is quadratically larger than optimal.

\subsection{New Result}
As mentioned before, there is no explicit construction of almost optimal $(n,(r,s))$-CFFs
when $d^{o(1)}<r<\omega(d/(\log\log d\log\log\log d)$
and the result of \cite{GLP} is quadratically larger than optimal.
In this paper we close this quadratic gap and give an explicit construction of an almost optimal $(n,(r,s))$-CFF
for all $r$ and $s$. Our main result is
\begin{theorem}\label{thm:main}
Fix any integers $r<s<d$ with $d=r+s$. There is an almost
optimal $(n,(r,s))$-CFF, i.e., of size
$$N(r,s)^{1+o(1)}\cdot \log n,$$ that can be constructed in linear time.
That is, in time

$$O(N(r,s)^{1+o(1)}\cdot n\cdot \log n)$$
\end{theorem}
As we've seen in Section~\ref{OR}, the above
theorem is already proved for $r<d^{o(1)}$ and $r>\omega(d/(\log\log d\log\log\log d)).$
\section{Applications of result}

\ignore{CFFs are related quite directly to non-adaptive learning of monotone DNFs with restricted clause width: Fix a DNF $\phi$ with $(s+1)$ terms $\phi=C_1\vee\ldots \vee C_{s+1}$ with each $C_i$ containing at most $r$ variables. Intuitively, an assignment that is true on all variables of $C_i$, and false on at least one variable in any other clause is a `witness' for $C_i$ being in the DNF.
For this reason, there has been significant interest in explicit constructions of them (cf. intro of
\cite{ABM15}).
In fact, our result was used recently by Abasi et. al. \cite{ABM15}
to give for the first time such a determinstic learning algorithm that is polynomial time
and makes an almost optimal number of queries.}
\subsection{Application to learning hypergraphs}
Let ${\cal G}_{s,r}$ be a set of all labeled hypergraphs of
rank at most $r$ (the maximum size of an edge $e\subseteq V$
in the hypergraph) on the set of vertices $V=\{1,2,\ldots,n\}$
with at most $s$ edges.
Given a hidden Sperner hypergraph\footnote{The hypergraph
is Sperner hypergraph if no edge is a
subset of another. If it is not Sperner hypergraph
then learning is not possible.} $G\in {\cal G}_{s,r}$, we need to identify it by asking
{\it edge-detecting queries}. An edge-detecting query $Q_G(S)$, for $S\subseteq V$
is: Does $S$ contain at least one edge of $G$? Our objective is to {\it non-adaptively} learn
the hypergraph $G$ by asking as few queries as possible.

This problem has many applications in chemical reactions, molecular biology and genome sequencing,
where deterministic non-adaptive algorithms are most desirable.
In chemical reactions, we are given a set of chemicals,
some of which react and some which do not.
When multiple chemicals are combined in one
test tube, a reaction is detectable if and
only if at least one set of the chemicals in the tube reacts.
The goal is to identify which sets react
using as few experiments as possible. The time needed
to compute which experiments to do is a
secondary consideration, though it is polynomial for
the algorithms we present.
See \cite{ABM15} and references within
for more details and many other applications in molecular biology.

The above hypergraph ${\cal G}_{s,r}$ learning problem is
equivalent to the
problem of exact learning a monotone DNF with at most $s$ monomials (monotone terms),
where each monomial contains at most $r$ variables
($s$-term $r$-MDNF) from membership queries~\cite{A87,AC08}.
A membership query, for an assignment $a\in \{0,1\}^n$ returns $f(a)$
where $f$ is the hidden $s$-term $r$-MDNF.

The non-adaptive learnability of $s$-term $r$-MDNF
was studied in~\cite{T99,MP04,MRY04,GHTW06,DH06,CLY13}.
All the algorithms are either deterministic algorithms
that uses non-optimal constructions of $(n,(s,r))$-CFF or randomized
algorithms that uses randomized constructions of $(n,(s,r))$-CFF.
Our construction in this paper gives, for the deterministic algorithm,
a better query complexity and changes the randomized algorithm to deterministic.
Recently, our construction is used in \cite{ABM15} to give
a polynomial time almost optimal algorithm for learning ${\cal G}_{s,r}$.

\subsection{Application to $r$-Simple $k$-Path}
Gabizon et. al. \cite{GLP} recently constructed deterministic algorithms for
parametrized problems with `relaxed disjointness constraints'.
For example, rather than searching for a \emph{simple} path of length $k$
in a graph of $n$ vertices, we can search for a path of length $k$ where no vertex is visited
more than $r$ times, for some `relaxation parameter' $r$.
We call the problem of deciding whether such a path exists \emph{\rsimp}.
Abasi et. al \cite{ABGH14} were the first to study \rsimp
and presented a randomized algorithm running in time
$O^*(r^{2k/r})$.
What is perhaps surprising, is that the running time can \emph{significantly improve} as $r$ grows.
Derandoming the result of \cite{ABGH14}, \cite{GLP} obtained a deterministic algorithm for \rsimp with running time
$O^*(r^{12k/r}\cdot 2^{O(k/r)})$.
At the core of their derandomization is the notion of a `multiset separator' -
a small family of `witnesses' for the fact that two multisets do not `intersect too much'
on any particular element.
How small this family of witnesses can be in turn depends on how small
an $(n,(2k/r,k-2k/r))$-CFF one can construct (details on these connections are given in Appendix \ref{app:paramalg}).
Plugging in our new construction into the machinery of \cite{GLP}, we
get
\begin{theorem}\label{thm:rsimpalg}
\emph{\rsimp} can be solved in deterministic time
$O(\repcomputenon \cdot k^{O(1)}\cdot  n^3\cdot \log n)$.
\end{theorem}
For example, when both $k/r$ and $r$ tend to infinity,
we get running time $O^*(r^{8k/r + o(k/r)})$ and \cite{GLP} get $O^*(r^{12k/r + o(k/r)})$.

In a well-known work, Koutis \cite{K08} observed that practically all parametrized problems
can be viewed as special cases of `multilinear monomial detection'.
\cite{GLP} also studied the relaxed version of this more general problem:
Given an arithmetic circuit $C$ computing an $n$-variate polynomial $f\in \Z[X_1,\ldots,X_n]$, determine
whether $f$ contains a monomial of total degree $k$ and individual degree at most $r$.
We call this problem \rmondetect.
\cite{GLP} define such a circuit $C$ to be \emph{non-canceling} if it contains only variables at its leaves (i.e., no constants),
and only addition and multiplication gates (i.e., no substractions).
\cite{GLP} showed that for non-canceling $C$,  \rmondetect can be solved in
time $O^*(|C|\cdot r^{18k/r}\cdot 2^{O(k/r)})$.
We obtain
\begin{theorem}\label{thm:rmonalg}
Given a non-canceling arithmetic circuit $C$ computing $f\in \Z[X_1,\ldots,X_n]$, \rmondetect can be solved in deterministic time
$O(|C|\cdot \mondetectcomputenon \cdot k^{O(1)}\cdot n^3\cdot \log n)$.
\end{theorem}

\subsection*{Organization of paper}
In Section \ref{sec:overview} we give an informal description
of our CFF construction.
In Section~\ref{First} we give a simple construction
that proves Theorem \ref{thm:main} for any $\log^2 d\le r\le d/(\log\log d)^{\omega(1)}$.
In Section~\ref{Second}, we give the proof for $d/(\log d)^{\omega(1)} \le r\le d/{\omega(1)}$.
The proofs of Theorems \ref{thm:rsimpalg} and \ref{thm:rmonalg}
appear in Appendix \ref{app:paramalg}.

\section{Proof Overview}\label{sec:overview}

Our construction is essentially a generalization of \cite{GLP} allowing a more flexible choices of parameters.
For simplicity, we first describe the construction of \cite{GLP}
and then explain our improvements.
%

To illustrate the ideas in a simple way, the following `adaptive' viewpoint will be
convenient:
We are given two disjoint subsets $C,D\subseteq [n]$ of sizes $|C|=r$ and $|D|=s$.
We wish to divide $[n]$ into two separate buckets such that all elements
of $C$ fall into the first, and all elements of $D$ fall into the second.
Of course the point in CFFs is \emph{that we do not know $C$ and $D$ in advance}.
However, the number of different possibilites for the division that will come up in the process will be a bound on the size of an analogous \cff - which will contain a vector $a\in \B^n$ corresponding to each way of separating $[n]$ into two buckets that came up in
the adpative process.

As a first step we use a perfect hash function $h$ to divide $[n]$ into $r$ buckets
such that each bucket contains exactly one element of $C$.
Using a construction of  Naor et. al \cite{NSS95},
$h$ can be chosen from a family of size $2^{O(r)}\cdot \log n$.
Let us call these buckets $B_1,\ldots,B_r$.
Now, suppose that we knew, for each $i\in [r]$, the number of elements $s_i$ from $D$
that fell into bucket $B_i$.
In that case we could use an $(n,(1,s_i))$-CFF $\cF_i$ to separate the element
of $C$ in $B_i$ from the $s_i$ elements of $D$, and put each in the correct final bucket.

We have such $\cF_i$ of size $c\cdot s_i^2\cdot \log n$ for universal constant $c$.
Thus, the number of different choices in all buckets is
\[\prod_{i=1}^r c\cdot s_i^2\cdot \log n \leq c^r\cdot (s/r)^{2r}\cdot \log^r n,  \]
as the product of the $s_i$'s is maximized when $s_1=\ldots s_r = s/r$.
Furthermore, \cite{GLP} show this can be improved to roughly $ (s/r)^r\cdot \log n
\leq (d/r)^r \cdot \log n$ where $d=r+s$.
This is done using the hitting sets for combinatorial rectangles of Linial et. al \cite{LLSZ97}
(we do not go into details on this stage here).
Of course, we do not know the $s_i$'s.
However, it is not too costly to simply guess them!
Or rather, try all options:
The number of choices for non-negative integers $s_1,\ldots,s_r$ such
that $s_1+\ldots + s_r = s$ is at most
\[\binom{d-1}{r-1}\leq \binom{d}{r} \leq (ed/r)^r.\]

Combining all stages, this gives us an \cff of size roughly
$(d/r)^{2r+O(1)}\cdot \log n$.
To get an almost optimal construction, we need to get the 2 in the exponent down to a 1.
We achieve this by reducing the cost of the `guessing stage'.
Instead of $r$ buckets, we begin by dividing $[n]$ into $k$ buckets
for some  $k=o(r)$,
such that every bucket will contain $r/k$ elements of $C$.
This is done using \emph{splitters} \cite{NSS95}.
For concreteness, think of $k=r/\log\log d$. (In the final construction we
need to choose $k$ more delicately).
Now as we only have $k$ $s_i$'s, there will be  less
possibilites to go over such that
$s_1+\ldots+ s_k =s$ - specifically less than $(ed/k)^k$.
On the other hand, our task in each bucket is now more costly -
we need to separate $r/k$ elements of $C$ from $s_i$ elements of $D$,
rather than just \emph{one} element of $C$.
A careful choice of parameters show this process can be done
while going over at most $(d/r)^{1+o(1)}$ options for the partition into two buckets.

\noindent
There are now two main technical issues left to deal with.
\begin{itemize}

 \item The splitter construction of \cite{NSS95} was not analyzed as being almost-linear time,
 but rather, only polynomial time.
 We give a more careful analysis of it's runtime.

 \item We need to generalize a component from the construction of  \cite{GLP},
 into what we call ``multi-CFFs''.
 Roughly speaking, this is a small set of strings of length $n\cdot \ell$
 that are `simultaneously a CFF on each $n$-bit block'.
 That is, if we think of the string as divided into $\ell$ blocks of length $n$,
 and wish to see in each block a certain pattern of weight $r_i$ in some subset of $d_i$ indices
 of that block, there will be one string in the multi-CFF that simultaneously exhibits all patterns.
 We construct a small multi-CFF using a combination of ``dense separating hash functions''
 and the hitting sets for combinatorial rectangles of \cite{LLSZ97}.
 See Section \ref{sec:secondConst} for details.

 \end{itemize}

\section{The First Construction}\label{First}
In this section we give the first construction

\subsection{Preliminary Results for the First Construction}
We begin by giving some definitions and preliminary results that
we will need for our first construction. The results in this subsection are from
\cite{NSS95} and \cite{B14b}.

Let $n,q$ and $d$ be integers.
Let $\cF$ be a set of boolean functions $f:[q]^d\to \{0,1\}$.
Let $H$ be a family of functions $h:[n]\to [q]$. We say that
$H$ is an $(n,\cF)$-{\it restriction family} ($(n,\cF)$-RF) if for every $\{i_1,\ldots,i_d\}\subseteq [n]$,
$1\le i_1<i_2<\cdots<i_d\le n$ and every $f\in \cF$
there is a function $h\in H$ such that $f(h(i_1),\ldots,h(i_d))=1$.

We say that a construction of an $(n,\cF)$-restriction family $H$
is a {\it linear time construction}, if it
runs in time $\tilde O(|H|\cdot n)=|H|\cdot n\cdot poly(\log |H|,\log n)$.

Let $H$ be a family of functions $h:[n]\to [q]$.
For $d\le q$ we say that $H$ is an $(n,q,d)$-{\it perfect hash family} ($(n,q,d)$-PHF)
if for every
subset $S\subseteq [n]$ of size $|S|=d$ there is a {\it hash
function} $h\in H$ such that $h|_S$ is injective (one-to-one) on $S$, i.e.,
$|h(S)|=d$. Obviously, an $(n,q,d)$-PHF is an $(n,\cF)$-RF when  $\cF=\{f\}$, for some $f:[q]^d\to \{0,1\}$
satisfying $f(\sigma_1,\ldots,\sigma_d)=1$ iff $\sigma_1,\ldots,\sigma_d$ are distinct.

In \cite{B14b} Bshouty proved
\begin{lemma}\label{ThH1c}\label{phf}
Let $q$ be a power of prime. If $q>4(d(d-1)/2+1)$ then
there is a linear time construction of an $(n,q,d)$-PHF
of size
$$O\left(\frac{d^2\log n}{\log(q/d^2)}\right).$$
\end{lemma}
The following is a folklore result
\begin{lemma}\label{com} Let $\cF$ be a set of boolean functions $f:[q]^d\to \{0,1\}$.
If there is a linear time construction of an $(m,\cF)$-RF where $m>4(d(d-1)/2+1)$
of size $s$ then there is a linear time construction of an $(n,\cF)$-RF
of size $$O\left(\frac{sd^2\log n}{\log(m/d^2)}\right).$$
\end{lemma}
\begin{proof} Let $H_1$ be an $(m,\cF)$-RF and let $H_2$ be the $(n,m,d)$-PHF
constructed in Lemma~\ref{phf}. Then it is easy to see that $H_1(H_2):=\{h_1(h_2)
\ |\ h_2\in H_2, h_1\in H_1\}$ is an $(n,\cF)$-RF.\qed
\end{proof}

Another restriction family that will be used here is splitters \cite{NSS95}.
An {\it $(n,r,k)$-splitter} is a family of functions $H$ from $[n]$ to $[k]$
such that for all $S\subseteq [n]$ with $|S|=r$, there is $h\in H$ that splits
$S$ perfectly, i.e., for all $j\in [k]$,
$|h^{-1}(j)\cap S|\in \{\lfloor r/k\rfloor,\lceil r/k\rceil\}$.
Obviously, an $(n,q,d)$-PHF is an $(n,d,q)$-splitter.
Define
\begin{eqnarray}\label{srke}
\sigma(r,k):=\left(\frac{2\pi r}{k}\right)^{k/2}e^{k^2/(12r)}.
\end{eqnarray}
From the union bound it can be shown that there exists an $(n,r,k)$-splitter
of size $O(\sqrt{r}\sigma(r,k)\log n)$, \cite{NSS95}. Naor et. al, \cite{NSS95}, use
the $r$-wise independent probability space to construct an $(m,r,k)$-splitter.
They show
\begin{lemma}\label{UB} For $k\le r$, an $(m,r,k)$-splitter of size
$O(\sqrt{r}\sigma(r,k)\log m)$ can be constructed in time
$$O\left(\sqrt{r}\cdot\sigma(r,k)m^{2r}\log m\right).$$
\end{lemma}
When $k=\omega(\sqrt{r})$, Naor et. al. in \cite{NSS95}, constructed an $(n,r,k)$-splitter
of size $O(\sigma(r,k)^{1+o(1)}\log n)$ in polynomial time. \blue{We here
show that the same construction can be done in \emph{linear time}}.
They first construct an $((r/z)^2,r/z,$ $k/z)$-splitter using Lemma~\ref{UB}
where $z=\Theta(r\log k/(k\log(2r/k)))$. They then use Lemma~\ref{com} to construct
an $(r^2,r/z,k/z)$-splitter. Then compose $z$ pieces of the latter to
construct an $(r^2,r,k)$-splitter and then again use Lemma~\ref{com}
to construct the final $(n,r,k)$-splitter.

Note here that we assume that $z|k|r$. The result can be extended to
any $z,k$ and $r$.

We now prove
\begin{lemma}\label{splitter} \label{Fc} For $k= \omega(\sqrt{r})$ and $z=16r\log k/(k\log(4r/k))$.
An $(n,r,k)$-splitter of size
$$r^{O(z)}\sigma(r,k)\log n=\sigma(r,k)^{1+o(1)}\log n$$ can be constructed in time $O(\sigma(r,k)^{1+o(1)}\log n)$.
\end{lemma}
\begin{proof} \blue{By Lemma~\ref{LL1} in Appendix \ref{app:splitters}, $z$ is a monotonic decreasing function
in $k$ and $16\sqrt{r}\ge z\ge 8\log r$ for $\sqrt{r}\le k\le r$.
First we construct an $((r/z)^2,r/z,$ $k/z)$-splitter using Lemma~\ref{UB}.
By Lemma~\ref{UB} and Lemma~\ref{LL2} in Appendix \ref{app:splitters}, this takes time
$$O(\sqrt{r/z}\cdot \sigma(r/z,k/z) ((r/z)^2)^{2r/z}\log (r/z))=o(\sigma(r,k)).$$
By Lemma~\ref{UB}, the size of this splitter is $O(\sqrt{r/z}\cdot \sigma(r/z,k/z)\log (r/z))$.
By Lemma~\ref{com}, using the above splitter,
an $(r^2,r/z,k/z)$-splitter $H$ of size
$$O((r/z)^{2.5} \sigma(r/z,k/z) \log (r/z) \log r)$$
can be constructed in linear time. Now, for every choice
of $0=i_0< i_1<i_2<\cdots<i_{z-1}<i_z={r^2}$ and $h_0,h_1,\ldots,h_{z-1}\in H$
define the function $h(j)=h_{t}(j)+(k/z)t$ if $i_t< j \le  i_{t+1}$.
It is easy to see that this gives an $(r^2,r,k)$-splitter. The splitter
can be constructed in linear time and by Lemma~\ref{LL3} in Appendix \ref{app:splitters}, its size
is
$${r^2\choose z-1}\left(c_1(r/z)^{2.5}\sigma(r/z,k/z)\log (r/z)\log r\right)^z=r^{c_2z}\sigma(r,k)$$ for some constants $c_1$ and $c_2$.
Now by Lemma~\ref{com} and Lemma~\ref{LL4} in Appendix \ref{app:splitters}, an $(n,r,k)$-splitter can be constructed in time
$$O(r^2(r^{c_2z}\sigma(r,k))\log n) =r^{O(z)}\sigma(r,k)\log n=\sigma(r,k)^{1+o(1)}\log n.$$}\qed
\end{proof}
The following is from \cite{B14b}
\begin{lemma}\label{CFFr}\label{cff}
There is an
$(n,(r,s))$-CFF of size
$$O\left(rs{2rs\choose r} \log n\right)$$ that can be constructed in linear time.
\end{lemma}
\subsection{Construction I}
Let $r\le s$ be integers and $d=r+s$. Obviously, $1\le r\le d/2$ and $d/2\le s\le d$.
We may also assume that
\begin{eqnarray}\label{assume}
r>poly(\log d)=d^{o(1)}.
\end{eqnarray} See the table in Section~\ref{OR}
and the discussion following it.

We first use Lemma~\ref{com} to reduce the problem to constructing
a $(q,(r,s))$-CFF for $q=O(d^3)$. We then do the following.
Suppose $1\le i_1<i_2<\cdots<i_d\le q$
and let $(\xi_1,\ldots,\xi_d)\in \{0,1\}^d$ with $r$ ones (and $s$ zeros)
that is supposed to be assigned to $(i_1,i_2,\cdots,i_d)$.
Let $i_{j_1},\ldots,i_{j_r}$ be the entries for which $\xi_{{j_1}},\ldots,\xi_{{j_r}}$
are equal to $1$. The main idea of the construction is to first deal with entries
$i_{j_1},\ldots,i_{j_r}$ that
are assigned to one and distribute them equally into $k$ buckets, where $k$ will
be determined later.
This can be done using a $(q,r,k)$-splitter.
Each bucket will contains $r/k$ ones and an unknown number of zeros.
We do not know how many zeros, say $d_i-(r/k)$, fall in bucket $i$  but we know
that $d_1+\cdots+d_k=d$. That is, bucket $i$ contains $d_i$ indices of $i_1,i_2,\cdots,i_d$
for which $r/k$ of them are ones. We take all possible $d_1+\cdots+d_k=d$
and for each bucket $i$ construct $(q,d_i-(r/k),r/k)$-CFF. Taking all possible
functions in each bucket for each possible $d_1+\cdots+d_k=d$ solves the problem.

Let $H_1$ be an $(n,q,d)$-PHF such that $d^3<q\le 2d^3$ is a power of prime
and $d=r+s$. The following follows from Lemma~\ref{com}
\begin{lemma} \label{HF} If $H$ is a $(q,(r,s))$-CFF then $\{h_1(h)\ | h_1\in H_1, h\in H\}$
is $(n,(r,s))$-CFF of size $|H|\cdot |H_1|$.
\end{lemma}

We now construct a $(q,(r,s))$-CFF.
Let $H_2$ be a $(q,r,k)$-splitter where $k<r$ will
be determined later. Let $H_3'[d']$ and $H_3''[d']$ be a $(q,d'-\lfloor r/k\rfloor,\lfloor r/k\rfloor)$-CFF
and $(q,d'-\lceil r/k\rceil,\lceil r/k\rceil)$-CFF respectively and define $H_3[d']:=H_3'[d']\cup H_3''[d']$
where $d\ge d'\ge \lceil r/k\rceil$.
For every $(h_1,\ldots,h_k)\in H_3[d_1]\times\cdots\times H_3[d_k]$ where
$d_1+\cdots+d_k=d$ and $g\in H_2$ define the function
$$H_{h_1,\ldots,h_k,g}(i)= h_{g(i)}(i).$$
We first prove
\begin{lemma}\label{FC}  The set of all $H_{h_1,\ldots,h_k,g}$ where
$(h_1,\ldots,h_k)\in H_3[d_1]\times\cdots\times H_3[d_k]$ for some
$d_1+\cdots+d_k=d$ and $g\in  H_2$ is a $(q,(r,s))$-CFF.
\end{lemma}
\begin{proof}  Consider any $1\le i_1<i_2<\cdots<i_d\le q$ and
any $(\xi_1,\ldots,\xi_d)$ of weight~$r$. Let $S=\{i_1,\ldots,i_d\}$.
Consider $I=\{ i_j\ |\ \xi_j=1\}$.
Since $H_2$ is a $(q,r,k)$-splitter there is $g\in H_2$ such that
$|g^{-1}(j)\cap I|\in \{ \lfloor r/k\rfloor, \lceil r/k\rceil\}$ for all $j=1,\ldots,k$.
Let $d_j= |g^{-1}(j)\cap S|$ for $j=1,\ldots,k$. Then
$d_1+d_2+\cdots+d_k=d$. Since $H_3[d_j]$ is a $(q,d_j-\lfloor r/k\rfloor,\lfloor r/k\rfloor)$-CFF
and $(q,d_j-\lceil r/k\rceil,\lceil r/k\rceil)$-CFF, there is $h_j\in H_3[d_j]$ such that
$h_j(g^{-1}(j)\cap I)=\{1\}$ and $h_j(g^{-1}(j)\cap (S\backslash I))=\{0\}$.

Now, if $\xi_\ell=1$ then $i_\ell\in I$. Suppose $g(i_\ell)=j$.
Then $i_\ell\in g^{-1}(j)\cap I$ and
$$H_{h_1,\ldots,h_k,g}(i_\ell)=h_j(i_\ell)\in h_j(g^{-1}(j)\cap I)=\{1\}.$$
If $\xi_\ell=0$ then $i_\ell\in S\backslash I$. Suppose $g(i_\ell)=j$.
Then $i_\ell\in g^{-1}(j)\cap (S\backslash I)$ and
$$H_{h_1,\ldots,h_k,g}(i_\ell)=h_j(i_\ell)\in h_j(g^{-1}(j)\cap (S\backslash I))=\{0\}.$$
\end{proof}

\subsection{Size of Construction I}
We now analyze the size of the construction. We will use $c_1,c_2,\ldots$ for constants
that are independent of $r,s$ and $n$.

Let $d^3<q\le 2d^3$ be a power of prime.
By Lemma~\ref{HF} and Lemma~\ref{FC} the size of the construction is
$$N:=|H_1|\cdot |H_2|\cdot \left| \bigcup_{d_1+\cdots+d_k=d} H_3[d_1]\times \cdots\times H_3[d_k] \right|$$
where $H_1$ is an $(n,q,d)$-PHF, $H_2$ is a $(q,r,k)$-splitter and $H_3[d']$ is a $(q,d'-\lceil r/k\rceil,\lceil r/k\rceil)$-CFF
and $(q,d'-\lfloor r/k\rfloor,\lfloor r/k\rfloor)$-CFF.

Let $z=16r\log k/(k\log(4r/k))$. By Lemma~\ref{splitter},\ref{phf} and \ref{cff} we have

\begin{eqnarray}
N&\le & c_1  \frac{d^2\log n}{\log d} \cdot r^{O(z)}\sigma(r,k)(\log d)\cdot \nonumber\\
&&\ \ \ \ \ \ \ \
\sum_{d_1+\cdots+d_k=d}\prod_{i=1}^k c_2\frac{d_ir}{k}{2d_i\lceil r/k\rceil\choose \lceil r/k\rceil}
\log d\nonumber\\
&\le& c_1d^2 r^{O(z)} \left(\frac{2\pi r}{k}\right)^{k/2} e^{k^2/(12r)}
(\log n)\cdot\nonumber\\
& & \ \ \ \ \ \ \ \
c_3^k\left(\frac{r\log d}{k}\right)^k\sum_{d_1+\cdots+d_k=d} \prod_{i=1}^k (2ed_i)^{r/k+1} d_i\label{St01}\\
&\le& c_4^k d^2 r^{O(z)}e^{k^2/(12r)} \left(\frac{r^3\log^2 d}{k^3}\right)^{k/2} (2e)^{r}
(\log n)\sum_{d_1+\cdots+d_k=d} \prod_{i=1}^k d_i^{r/k+2}\nonumber \\
&\le& c_5^k d^2 r^{O(z)}e^{k^2/(12r)}\left(\frac{r^3\log^2 d}{k^3}\right)^{k/2} (2e)^{r}
(\log n)\left(\frac{d}{k}\right)^k \max_{d_1+\cdots+d_k=d} \left(\prod_{i=1}^k d_i\right)^{r/k+2}
\label{St02} \\
&\le& c_6^k d^2  r^{O(z)}e^{k^2/(12r)}\left(\frac{r^3\log^2 d}{k^3}\right)^{k/2} (2e)^{r}
 \left(\frac{d}{k}\right)^{r+3k} \log n\label{St03}\\
&\le& c_6^k d^2 r^{O(z)}e^{k^2/(12r)}\left(\frac{r^3d^6\log^2 d}{k^9}\right)^{k/2} \left(\frac{2er}{k}\right)^r\left(\frac{d}{r}\right)^{r} \log n\nonumber
\end{eqnarray}
\blue{
(\ref{St01}) follows from (\ref{srke}) and the fact that ${a\choose b}\le (ea/b)^b$.
(\ref{St02}) follows from the fact that the number of $k$-tuples $(d_1,\ldots,d_k)$ such that
$d_1+\cdots+d_k=d$ is ${d+k-1\choose k-1}\le c^k(d/k)^k$ for some constant $c$.
(\ref{St03}) follows from the fact that $\max_{d_1+\cdots+d_k=d} \prod_{i=1}^k d_i=(d/k)^k$.
}

In summary, we have
$$N\le c_6^k d^2 r^{O(z)}e^{k^2/(12r)}\left(\frac{r^3d^6\log^2 d}{k^9}\right)^{k/2} \left(\frac{2er}{k}\right)^r\left(\frac{d}{r}\right)^{r} \log n.$$
Now assume $r>\log^2 d$ (see (\ref{assume})) and let $k:=r/\log\log d$.

Since
$$z\log r=\frac{16r\log k\log r}{k\log(4r/k)}\le c_7\frac{\log^2 r\log\log d}{\log\log\log d}=o(r),$$
$$\frac{k^2}{12r}=\frac{r}{12(\log\log d)^2}=o(r)$$
$$\left(\frac{r^3d^6\log^2 d}{k^9}\right)^{k/2}=c_8^r\left(\frac{d}{r}\right)^{3k}=c_8^r\left(\frac{d}{r}\right)^{o(r)},$$
and $d/r \geq 2$, we have,
$$N\le (c_9\log\log d)^r \left(\frac{d}{r}\right)^{r(1+o(1))} \log n.$$

This is
$$\left(\frac{d}{r}\right)^{r(1+o(1))} \log n= N(r,s)^{1+o(1)}\log n$$ when $$\log^2d\le r\le\frac{d}{(\log\log d)^{\omega(1)}}.
$$

\section{The Second Construction}\label{sec:secondConst}\label{Second}
In the second construction we replace each component $H_3[d_1]\times \cdots \times H_3[d_k]$ with
another construction that is built from scratch and therefore has smaller size.
The main idea is the following: rather than taking all possible functions
in each $(q,d_i-(r/k),r/k)$-CFF in each bucket, we construct what we call a ``multi-CFF''.
\blue{We first construct a dense ``separating hash family'' that maps the
entries to a smaller domain $[q]$ and separates entries that are supposed to
be assigned zero from those that are suppose to be assigned one (i.e.,
they are mapped to disjoint sets).  This is done in each bucket.
We then use the hitting set for dense combinatorial rectangles of Linial et. al, \cite{LLSZ97}, to give
a separating hash family for all the buckets. Then we build a multi-CFF by assigning
$0$ and $1$ to every possible two disjoint sets.}
We proceed with the details of the second construction.

\subsection{Preliminary Results For the Second Construction}

Let $H$ be a set of functions $h:[n]\to [q]$.
We say that $H$ is a $(1-\epsilon)$-{\it dense} $(n,q,(\rho_1,\rho_2))$-Separating Hash Family
(SHF) if for every
two disjoint subsets $S_1,S_2\subseteq [n]$ of sizes $|S_1|=\rho_1,|S_2|=\rho_2$ there are at least
$(1-\epsilon)|H|$ hash
functions $h\in H$ such that $h(S_1)\cap h(S_2)=\O$.

The following lemma follows from~\cite{B14b}.
\begin{lemma}~\label{Den} Let $q$ be a power of prime. If $\epsilon>4(\rho_1\rho_2+1)/q$ then
there is a $(1-\epsilon)$-dense $(n,q,(\rho_1,\rho_2))$-SHF
of size
$$O\left(\frac{\rho_1\rho_2\log n}{\epsilon\log(\epsilon q/e(\rho_1\rho_2+1))}\right)$$
that can be constructed in linear time.
\end{lemma}

Let $R\subseteq [t]^k$ be
a set of the form $R_1\times \ldots\times R_k$, where
$R_i\subseteq [t]$.
We say $R$ is a {\it combinatorial rectangle with sidewise density $\gamma$},
if for every $i\in [t]$, $|R_i|\geq \gamma \cdot t$.
A set $H\subseteq [t]^k$ is called a \emph{hitting set for rectangles with sidewise density $\gamma$}
if for every set $R\subseteq  [t]^k$
that is a combinatorial rectangle of sidewise density $\gamma$,
$R\cap H \neq \emptyset$.

Linial et. al \cite{LLSZ97} gave the following construction of a hitting set for combinatorial rectangles.
\begin{lemma}\label{thm:LLSZhittingset}
A hitting set  for rectangles $H\subseteq [t]^k$ with sidewise density $1/3$
of size $|H|= t^{O(1)}\cdot  2^{O(k)}$ can be constructed
in time $t^{O(1)}\cdot  2^{O(k)}$.
\end{lemma}

Let $H$ be a set of functions $h:[k]\times[n]\to \{0,1\}$.
We say that $H$ is an
$(n,((\rho_{1,1},\rho_{1,2}),\ldots,(\rho_{k,1},\rho_{k,2})))$-Multi-CFF (MCFF)
if for every $k$ pairs
of disjoint subsets $(S_{i,1},S_{i,2})\subseteq [n]$ of sizes $|S_{i,1}|=\rho_{i,1},|S_{i,2}|=\rho_{i,2}$,
$i=1,\ldots,k$, there is $h\in H$ such that
$h(i,S_{i,1})=1$ and $h(i,S_{i,2})=0$ for all $i=1,\ldots,k$.

We now prove
\begin{lemma} \label{MCFF} There is an $(n,((\rho_{1,1},\rho_{1,2}),\ldots,(\rho_{k,1},\rho_{k,2})))$-MCFF
of size
$$( 2^k(\log n)\max_i \rho_{i,1}\rho_{i,2})^{O(1)} \prod_{i=1}^k {48\rho_{i,1}\rho_{i,2}\choose \rho_{i,1}}$$
that can be constructed in time $n\times poly((\max_i\rho_{i,1}\rho_{i,2})2^k\log n)$
\end{lemma}
\begin{proof} We first choose integers $q_i$, $i=1,\ldots,k$ that are
powers of primes $24 \rho_{i,1}\rho_{i,2}< q_i\le 48 \rho_{i,1}\rho_{i,2}$.
Since $4(\rho_{i,1}\rho_{i,2})/q_i<1/2$, by Lemma~\ref{Den}, there is a $1/2$-dense $(n,q_i,(\rho_{i,1}\rho_{i,2}))$-SHF $H_i$
of size $|H_i|=t=O((\max_i $ $\rho_{i,1}\rho_{i,2}) (\log n))$. Let $H_i=\{h_{i,1},\ldots,h_{i,t}\}$.
Let $G\subseteq [t]^k$ be a hitting set for rectangles with sidewise density $1/3$
of size $|G|= t^{O(1)}\cdot  2^{O(k)}$. By Lemma~\ref{thm:LLSZhittingset} this set
can be constructed in time $t^{O(1)}\cdot  2^{O(k)}=poly((\max_i\rho_{i,1}\rho_{i,2})2^k\log n).$

Now for every $g\in G$ and every $R_i\subset [q_i]$,
of size $|R_i|= \rho_{i,1}$, $i=1,\ldots,k$,
consider the functions $h_{1,g_1},h_{2,g_2},\ldots,h_{t,g_t}$
and define $h:[k]\times [n]\to \{0,1\}$ as follows: $h(i,j)=1$
iff $h_{i,g_i}(j)\in R_i$.

To show that the set of all such $h$ is an $(n,((\rho_{1,1},\rho_{1,2}),\ldots,(\rho_{k,1},\rho_{k,2}))$-MCFF,
consider $k$ pairs
of disjoint subsets $(S_{i,1},S_{i,2})\subseteq [n]$ of sizes $|S_{i,1}|=\rho_{i,1},|S_{i,2}|=\rho_{i,2}$,
$i=1,\ldots,k$. Let $H_i^*=\{h'\in H_i\ |\ h'(S_{i,1})\cap h'(S_{i,2})=\O\}$.
Since $H_i$ is a $1/2$-dense $(n,q_i,(\rho_{i,1}\rho_{i,2}))$-SHF, we have $|H_i^*|\ge |H_i|/2$.
Since $G\subseteq [t]^k$ is a hitting set for rectangles with sidewise density $1/3$
there is $g\in G$ such that $h_{i,g_i}\in H_i^*$ for all $i=1,\ldots, k$. Let $R_i$
be any set of size $\rho_{i,1}$ such that $h_{i,g_i}(S_{i,1})\subseteq R_i\subseteq [q_i]\backslash h_{i,g_i}(S_{i,2})$.
Then the function $h$ defined above satisfies {the following}: since $h_{i,g_i}(S_{i,1})\subseteq R_i$ we have $h(i,S_{i,1})={1}$
and since $R_i\cap h_{i,g_i}(S_{i,2})=\O$ we have $h(i,S_{i,2})={0}$ for all $i=1,\ldots,k$.

The number of such functions $h$ is
$$|G|\prod_{i=1}^k {q_i\choose |R_i|}.$$
\end{proof}
\subsection{Analysis for Construction II}
In the analysis we just replace the size of $H_3[d_1]\times \cdots \times H_3[d_k]$ in the analysis
of construction I to the new size of a $(q,((d_1-r/k,r/k),\ldots,(d_k-r/k,r/k)))$-MCFF in Lemma~\ref{MCFF}
where $d^3<q\le 2d^3$ and get
\begin{eqnarray*}
N&\le & c_1  \frac{d^2\log n}{\log d} \cdot r^{O(z)}\sigma(r,k)(\log d)\cdot \\
&&\sum_{d_1+\cdots+d_k=d} 2^{O(k)}(\log d)^{O(1)}\left(\frac{dr}{k}\right)^{O(1)}\prod_{i=1}^k{c_2d_i\lceil r/k\rceil\choose \lceil r/k\rceil}\\
&\le& c_3^k d^{O(1)} r^{O(z)} \left(\frac{2\pi r}{k}\right)^{k/2} e^{k^2/(12r)}
\log n\sum_{d_1+\cdots+d_k=d} \prod_{i=1}^k (c_4d_i)^{r/k+1} \\
&\le& c_4^k d^{O(1)} r^{O(z)}e^{k^2/(12r)} \left(\frac{r}{k}\right)^{k/2} c_5^{r}
\log n\sum_{d_1+\cdots+d_k=d} \prod_{i=1}^k d_i^{r/k+1} \\
&\le& c_6^k d^{O(1)} r^{O(z)}e^{k^2/(12r)}\left(\frac{r}{k}\right)^{k/2} c_5^{r}
\log n\left(\frac{d}{k}\right)^k \max_{d_1+\cdots+d_k=d} \left(\prod_{i=1}^k d_i\right)^{r/k+1} \\
&\le& c_6^k d^{O(1)}  r^{O(z)}e^{k^2/(12r)}\left(\frac{r}{k}\right)^{k/2} c_5^{r}
 \left(\frac{d}{k}\right)^{r+2k} \log n\\
&\le& c_6^k d^{O(1)} r^{O(z)}e^{k^2/(12r)}\left(\frac{r}{k}\right)^{k/2} c_5^{r}\left(\frac{r}{k}\right)^{r+2k}\left(\frac{d}{r}\right)^{r+2k} \log n\\
\end{eqnarray*}

Now let $r>\log^2 d$ and $k=r/\varphi(d)$ where $\varphi(d)<\log d$ and $\varphi(d)=\omega(1)$. Then $k=o(r)$
$$c_4^kd^{O(1)}e^{k^2/(12r)}\left(\frac{r}{k}\right)^{k/2}=2^{O\left(\frac{r\log\varphi(d)}{\varphi(d)}\right)}=2^{o(r)}$$
and $$r^{O(z)}=r^{O((r/k)\log k/\log(2r/k))}=2^{O(\varphi(d)\log d/\log\varphi(d))}=2^{o(r)}.$$
Therefore
$$N=(c_7 \varphi(d))^{r+o(r)}\left(\frac{d}{r}\right)^{r+o(r)}{\log n}$$ which is
$$N(r,s)^{r(1+o(1))}\log n$$
when
$$r=\frac{d}{\varphi(d)^{\omega(1)}}.$$ Since $\varphi(d)<\log d$ is any function that
satisfies $\omega(1)$, the above is true for any
$$\frac{d}{(\log d)^{\omega(1)}}\le r\le \frac{d}{\omega(1)}.$$

\appendix

\section{Technical results for the proof of Lemma \ref{splitter}}\label{app:splitters}
In this appendix we give some proofs of technical results needed for Lemma \ref{splitter}.

Here we assume that $r$ and $k$ are large enough integers

\begin{lemma}\label{LL1} Let $z=16r\log k/(k\log (4r/k))$. Then $z$ is a monotonically
decreasing function in $k$ in the interval $[\sqrt{r},r]$. In particular, $16\sqrt{r}\ge z\ge 8\log r$.
\end{lemma}
\begin{proof} From $\partial z/\partial k|_{k=x}=0$ we get $\ln^2 x-(\ln 4r)\ln x +\ln (4r)=0$.
This gives two solutions $x_0,x_1$ for $x$. One satisfies $\ln x_0>\ln r$ and therefore $x_0>r$
and the second $\ln x_1<2$ and therefore $x_1<e^2<\sqrt{r}$. This implies that
the function is monotone in the interval $[\sqrt{r},r]$. Now since $z|_{k=\sqrt{r}}=16\sqrt{r}$
and $z|_{k=r}=8\log r$ the result follows.
\end{proof}

We remind the reader that
$$\sigma(r,k):=\left(\frac{2\pi r}{k}\right)^{k/2}e^{k^2/(12r)}.$$
\begin{lemma}\label{LL2} Let $z=16r\log k/(k\log (4r/k))>8$ and $k=\omega(\sqrt{r})$.
$$\sqrt{\frac{r}{z}}\cdot \sigma\left(\frac{r}{z},\frac{k}{z}\right)
\left(\frac{r}{z}\right)^{4r/z}\log \frac{r}{z}=o(\sigma(r,k)).$$
\end{lemma}
\begin{proof} First
\begin{eqnarray}\label{fireq}
\sigma\left(\frac{r}{z},\frac{k}{z}\right)
= \sigma({r},{k})^{1/z}\le \sigma({r},{k})^{\frac{1}{8}}.
\end{eqnarray}
Now
\begin{eqnarray*}
\sqrt{\frac{r}{z}}\cdot
\left(\frac{r}{z}\right)^{4r/z}\log \frac{r}{z}&\le &\left(\frac{r}{z}\right)^{5r/z}\\
&=& \left(\frac{k\log \frac{4r}{k}}{16\log k}\right)^{5r/z}\\
&\le& k^{5r/z}\\
&=& k^{\frac{5k\log \frac{4r}{k}}{16\log k}}\\
&=& \left(\frac{4r}{k}\right)^\frac{5k}{16}\le \sigma(r,k)^{\frac{5}{8}}.
\end{eqnarray*}
This with (\ref{fireq}) implies the result.
\end{proof}

\begin{lemma}\label{LL3} Let $z=16r\log k/(k\log (4r/k))$. Then
$${r^2\choose z-1}\left(c_1\left(\frac{r}{z}\right)^{2.5}
\sigma\left(\frac{r}{z},\frac{k}{z}\right)\log \frac{r}{z}\log r\right)^z=r^{O(z)}\cdot \sigma(r,k)$$
\end{lemma}
\begin{proof} First we have
$$\sigma\left(\frac{r}{z},\frac{k}{z}\right)^z=\sigma({r},{k}).$$
Now
\begin{eqnarray*}
{r^2\choose z-1}\left(c_1\left(\frac{r}{z}\right)^{2.5}
\log \frac{r}{z}\log r\right)^z&\le & \left(\frac{er^2}{z}\right)^z r^{4.5z}\le r^{7z}.
\end{eqnarray*}
\end{proof}

\begin{lemma}\label{LL4} Let $z=16r\log k/(k\log (4r/k))$. For $r\ge k=\omega (\sqrt{r})$ we have
$$r^{O(z)}=\sigma(r,k)^{o(1)}$$
\end{lemma}
\begin{proof} Let $k=\sqrt{r}\cdot \phi(r)$ where $\phi(r)=\omega(1)$. Then
for a constant $c$ there is a constant $c'$ such that
$$\log r^{cz}\le  {c'\frac{\sqrt{r}\log^2 r}{\phi(r)\log({\sqrt{r}/\phi(r)})}}$$ and there is a constant $c''$
such that
$$\log \sigma(r,k)\ge c''\phi(r)\sqrt{r}\log\frac{\sqrt{r}}{\phi(r)}.$$
Now, for some constant $c'''$,
$$\frac{\log r^{cz}}{\log \sigma(r,k)}\le c''' \frac{\log^2 r}{\phi^2(r)\log^2(\sqrt{r}/\phi(r))}=o(1).$$
\end{proof}

\ignore{
\begin{eqnarray*}
N&\le & c_1  \frac{d^2\log n}{\log d} \cdot c_2 r^2{r^2\choose k}(\log d)\cdot \\
&&\sum_{d_1+\cdots+d_k=d} c_3^k\prod_{i=1}^k \frac{d_ir}{k}{2d_i\lceil r/k\rceil\choose \lceil r/k\rceil} \log d\\
&\le& c_4^k (dr)^2 \left(\frac{r^2}{k}\right)^k \left(\frac{r\log d}{k}\right)^k
\log n\sum_{d_1+\cdots+d_k=d} \prod_{i=1}^k (2d_i)^{r/k+1} d_i\\
&\le& c_5^k (dr)^2 \left(\frac{r^3\log d}{k^2}\right)^k 2^{r}
\log n\sum_{d_1+\cdots+d_k=d} \prod_{i=1}^k d_i^{r/k+2} \\
&\le& c_6^k (dr)^2 \left(\frac{r^3\log d}{k^2}\right)^k 2^{r}
\log n\left(\frac{d}{k}\right)^k \max_{d_1+\cdots+d_k=d} \left(\prod_{i=1}^k d_i\right)^{r/k+2} \\
&\le& c_6^k (dr)^2  \left(\frac{r^3d\log d}{k^3}\right)^k 2^{r}
 \left(\frac{d}{k}\right)^{r+2k} \log n\\
&\le& c_6^k (dr)^2 \left(\frac{r^3d^3\log d}{k^5}\right)^k \left(\frac{2r}{k}\right)^r\left(\frac{d}{r}\right)^{r} \log n\\
\end{eqnarray*}
For $r>2\log d$ we take $k=r/\log d$ and get
$$c_7^{r/\log d} c_8^r (2\log d)^r \left(\frac{d}{r}\right)^{r} \log n=(c_9\log d)^r \left(\frac{d}{r}\right)^{r} \log n.$$
This is
$$\left(\frac{d}{r}\right)^{r(1+o(1))} \log n$$ when $$r=\frac{d}{(\log d)^{\omega(1)}}.$$

--------------------------------}

\renewcommand{\O}{O_k}
\section{Application to parametrized algorithms with relaxed disjointness constraints}\label{app:paramalg}
In this appendix, for the purpose of deriving Theorems \ref{thm:rsimpalg} and \ref{thm:rmonalg}, we explain how objects related to cover-free families were used by \cite{GLP} to obtain certain parameterized algortihms.

\paragraph*{Notation.}
Throughout this appendix, we use the notation $\O$ to hide $k^{O(1)}$ terms. We denote $[n]=\{1,2,\ldots,n\}$. For sets $A$ and $B$, by $\oftypenon{A}{B}$ we denote the set of all functions from $A$ to $B$. 
The notation $\triangleq$ is used to introduce new objects defined by formulas on the right hand side.



In fact, \cite{GLP} do not use CFFs directly, but related objects called \emph{minimal separating families} (Definition \ref{dfn:minseparating}) that have an additional injectivity property.
We begin by formally showing that CFFs indeed imply minimal separting families of similar size.

\subsection{From CFFs to minimal separating families}

\paragraph*{Hashing families.}
Recall that, for an integer $t\geq 1$, we say that a family of functions $\H\oftype{[n]}{[m]}$ is a \emph{$t$-perfect hash family},
if for every $C\subseteq [n]$ of size $|C| = t$ there is $f\in \H$ that is injective on $T$. Alon, Yuster and Zwick \cite{AYZ08} used a construction of Moni Naor
(based on ideas from Naor et al.~\cite{NSS95}) to hash a subset of size $t$ into a world of size $t^2$ using a very small set of functions:
\begin{theorem}[\cite{AYZ08} based on Naor]\label{thm:hash_to_k^2}
For integers $1\leq t\leq n$, a $t$-perfect hash family $\H\oftype{[n]}{[t^2]}$
of size $t^{O(1)}\cdot \log n$ can be constructed in time $O(t^{O(1)}\cdot n\cdot \log n)$
\end{theorem}

We will also use the following perfect hash family given by Naor, Schulman and Srinivasan \cite{NSS95}.
\begin{theorem}[\cite{NSS95}]\label{thm:NSShash}
For integers $1\leq t\leq n$, a $t$-perfect hash family $\H\oftype{[k^2]}{[t]}$ of size $e^{t+ O(\log^2 t)}\cdot \log k$
can be constructed in time $O(e^{t+ O(\log^2 t)}\cdot k\cdot \log k)$.
\end{theorem}

\begin{definition}[Minimal separating family]\label{dfn:minseparating}
A family of functions $\H \oftype{[n]}{[t+1]}$
is \emph{$(t,k)$-minimal separating} if
for every disjoint subsets $C,D\subseteq [n]$
with $|C| =  t $ and $|D| \leq  k-t$, there is a function $h\in \H$ such that
\begin{itemize}
\item $h(C)=[t]$.
\item $h(D)\subseteq \{t+1\}$.
\end{itemize}
\end{definition}

We show that small cover-free families imply small minimal-separting families.
\begin{lemma}\label{lem:cff->minimalsepfamily}
Fix any $t\leq k \leq n$.
Suppose a $(k^2,(t,k-t))$-CFF $\cF$ can be constructed in time $S$.
Then a $(t,k)$-minimal separating family of size $\O(|\cF|\cdot 2^{O(t)}\cdot \log n)$
Can be constructed in time $\O(S\cdot 2^{O(t)}\cdot \log n \cdot n)$.
\end{lemma}
\begin{proof}
Fix disjoint subsets $C,D\subseteq [n]$ with $|C|=t$ and $|D|\leq k-t$.
It will be convenient to present the family by constructing $h$
adaptively given $C$ and $D$. That is, for arbitrarily chosen $C$ and $D$, we will adaptively construct a function $h$ that separates $C$ from $D$. Function $h$ will be constructed by taking a number of {\em{choices}}, where each choice is taken among a number of possibilities. The final family $\H$ will comprise all $h$ that can be obtained using any such sequence of choices; thus, the product of the numbers of possibilities will limit the size of $\H$. As $C$ and $D$ are taken arbitrarily, it immediately follows that such $\H$ separates every pair $(C,D)$.

\begin{enumerate}
\item Let $\H_0\oftype{[n]}{[k^2]}$ be the $k$-perfect hash family given by
Theorem~\ref{thm:hash_to_k^2}. Choose $f_0 \in \H_0$ that is injective on $C\cup D$ --- there are $k^{O(1)}\cdot \log n$ choices for this stage.

From now on, we identify $C$ and $D$ with their images in $[k^2]$ under $f_0$.
\item Note that an element $f\in \cF$ can be viewed as a function $f:[k^2]\to \set{0,1}$. Now choose an element $f_1$ of the $(k^2,(t,k-t))$-CFF $\cF$, with $f_1(C)\equiv 1$ and $f_1(D)\equiv 0$ --- there are $|\cF|$ choices for this stage.

At this stage we have `separated' $C$ from $D$, and just need to satisfy the additional requirement of being injective on $C$.
\item  Let $\H_2\oftype{f_1^{-1}(1)}{[t]}$ be the $t$-perfect hash family given by Theorem~\ref{thm:NSShash}. Choose a function $f_2\in \H_2$ that is injective on $C$ --- there are $e^{t+ O(\log^2 t)}\cdot \log k$ choices for this stage.
\end{enumerate}
The running times and family size are immediate from the construction.
\end{proof}

Plugging in our construction from Theorem \ref{thm:main} to the above we get
\begin{corollary}\label{cor:sepfamily}
 Fix any $t\leq k \leq n$.
 A $(t,k)$-minimal separating family of size $\O( (k/t)^{t+o(t)}\cdot 2^{O(t)}\cdot \log n)$
can be constructed in time $\O((k/t)^{t+o(t)}\cdot 2^{O(t)}\cdot \log n\cdot n)$
\end{corollary}
\begin{proof}
 It's a straightforward plugin of Theorem \ref{thm:main} into Lemma \ref{lem:cff->minimalsepfamily}.
 The only thing to notice is that for any $t\leq k$,
 \[N(t,k-t)^{1+o(1)}\cdot 2^{O(t)} \leq  (k\cdot (ek/t)^t)^{1+o(1)}=
 \O( (k/t)^{t+o(t)}\cdot 2^{O(t)}).\]
\end{proof}

We proceed to define and construct \emph{multiset separators} that are smaller
than those in \cite{GLP}.

\newcommand{\rkcomp}{$(r,k)$-compatible\xspace}
\newcommand{\rkcons}{$(r,k)$-consistent\xspace}
\subsection{Multiset Separators}
\paragraph*{Notation for multisets.}
Fix integers integers $n,r,k\geq 1$.
We use $\rr$ to denote $\set{0,\ldots,r}$.
An \emph{\rset} is a multiset $A$ where each element of $[n]$
appears at most $r$ times.
It will be convenient to think of $A$ as a vector in $\rr^n$,
where $A_i$ denotes the number of times $i$ appears in $A$.
We denote by $|A|$ the number of elements in $A$ counting repetitions.
That is, $|A| = \sum_{i=1}^n A_i $.
We refer to $|A|$ as the \emph{size} of $A$.
An \emph{\rkset} is an \rset $A\in \rr^n$,
where the number of elements with repetitions is at most $k$.
That is, $|A| \leq k$. For two multisets $A,B$ over $[n]$,

Fix \rsets $A,B\in \rr^n$.
We say that $A\leq B$ when
$A_i\leq B_i$ for all $i\in [n]$.
By $\comp{A}\in \rr^n$ we denote
the ``complement'' of \rset $A$,
that is, $\comp{A}_i = r-A_i$ for all $i\in [n]$. By $A+B$ we denote the ``union'' of $A$ and $B$, that is, $(A+B)_i=A_i+B_i$ for all $i\in [n]$.
Suppose now that $A$ and $B$ are \rksets.
We say that $A$ and $B$ are \emph{\rkcomp}
if $A+B$ is also an \rkset, and $|A+B|=k$.
That is, the total number of elements with repetitions in $A$ and $B$ together is $k$
and any specific element $i\in [n]$ appears in $A$ and $B$ together at most $r$ times.
With the notation above at hand, we can define the central object needed for the algorithms of \cite{GLP}.
\begin{definition}[Multiset separator]\label{dfn:multiseparator}
Let $\F$ be a family of \rsets.
We say that $\F$ is an \rksep if
for any \rksets $A,B \in \rr^n$
that are \rkcomp, there exists $F\in \F$
such that $A\leq  F\leq \comp{B}$.
\end{definition}

\cite{GLP} showed that a minimal separating family can be used to construct an \rksep.
\begin{theorem}\label{thm:multisetsepmain}[\cite{GLP} Theorem 3.3]
Fix integers $n,r,k$ such that $1<r\leq k\leq n$,
and let $t \triangleq  \lfloor 2k/r \rfloor$.
Suppose a $(t,k)$-minimal separating family $\H\oftype{[n]}{[t+1]}$
can be constructed in time $f(r,k,n)$.
Then an \rksep $\F$ of size $|\H|\cdot (r+1)^t$
can be constructed in time $\O(f(r,k,t))\cdot (r+1)^t)$.
\end{theorem}

Plugging in our construction of minimal seperating families from Corollary \ref{cor:sepfamily} we get
\begin{corollary}\label{cor:rksep}
 Fix integers $n,r,k$ such that $1< r\leq k$. Then an \rksep $\F$ of size $\O(r^{4k/r+o(k/r)}\cdot 2^{O(k/r)}\cdot \log n)$
can be constructed in time $\O(r^{4k/r+o(k/r)}\cdot 2^{O(k/r)}\cdot n\cdot \log n)$
\end{corollary}
The above corollary is an analog of Corollary 3.4 in \cite{GLP}
where the exponent of $r$ was $6k/r$ rather than $4k/r  + o(k/r)$.

From this point on we do not give full details, as our theorems follow by a direct
plug in of Corollary \ref{cor:rksep} in \cite{GLP} as a relpacement for their Corollary 3.4.

Specfically, using Corollary \ref{cor:rksep}, the algorithm in Corollary 3.8
of \cite{GLP} for finding a represntative set of a family of multisets $\P$ will run in time $\O(|\P| \cdot \sepcompute)$ rather than  $\O(|\P|\cdot r^{6k/r}\cdot 2^{O(k/r)}\cdot n\log n)$ which will translate to the running times stated in Theorems \ref{thm:rsimpalg}
and \ref{thm:rmonalg} when running the Algorithms proving Theorems 5.6 and 5.8 in \cite{GLP}.

\end{document}